\documentclass[conference,a4paper]{IEEEtran}

\usepackage{preamble}
\usepackage{algorithm,algorithmic}

\usetikzlibrary{arrows}
\newcommand{\midarrow}{\tikz \draw[-Stealth] (0,0) -- +(.1,0);}

\usepackage{float}
\newfloat{algorithm}{t}{lop}

\begin{document}
\title{Decomposition of Clifford Gates}

\author{
  \IEEEauthorblockN{Tefjol Pllaha, Kalle Volanto, Olav Tirkkonen}
	
 \IEEEauthorblockA{\small Aalto University, Helsinki, Finland\\ e-mails: \{tefjol.pllaha, kalle.volanto, olav.tirkkonen\}@aalto.fi
    }
}

\maketitle

\begin{abstract}
In fault-tolerant quantum computation and quantum error-correction one is interested on Pauli matrices that commute with a circuit/unitary. 
%In this paper we focus on Clifford/stabilizer circuits.
We provide a fast algorithm that decomposes any Clifford gate as a~\emph{minimal} product of Clifford transvections. The algorithm can be directly used for finding all Pauli matrices that commute with any given Clifford gate. 
To achieve this goal, we exploit the structure of the symplectic group with a novel graphical approach.
\end{abstract}

\section{Introduction}
The Clifford group is of central importance in quantum information and computation.
This paper is primarily motivated by its importance in fault-tolerant quantum computation and quantum error-correction~\cite{Gottesman-phd97,CRSS98}. 
Traditionally, the Clifford group is studied via its connection with the binary symplectic group~\cite{CRSS97} and the associated decompositions of the latter. 
The Bruhat decomposition of the symplectic group~\cite{MR18} gives a standard generating set made of qubit permutation~\eqref{e-GDP}, diagonal gates~\eqref{e-GUS}, and partial Hadamard gates~\eqref{e-GO}. Alternatively, the Clifford group can be studied via the transvection decomposition of the symplectic group, which we briefly describe in Section~\ref{Sec-decomp}. 
It is well-known~\cite{OMeara,Callan} that the symplectic group is generated by symplectic transvections~\eqref{e-transvectionDef}. 
Although these references give a constructive proof, the decomposition primarily relies on exhaustive search. 
In this paper we give a simple and fast algorithm that decomposes any symplectic matrix as a~\emph{minimal} product of symplectic transvections. 
On the other hand, the Clifford gates~\eqref{e-transvection} correspond to symplectic transvections, and for this reason we will refer to them as Clifford transvections.
By definition, Clifford transvections are sparse (in fact, they are the most sparse Cliffords other than Paulis and diagonal Cliffords), and given their simple conjugation action, they are also easy to implement.
This yields directly a decomposition of any $m$-qubit Clifford gate as a minimal product of Clifford transvections.

We exploit the structure of symplectic matrices with a novel graphical approach. We associate to a symplectic matrix, written as a minimal product of transvections, a (binary, symmetric) Gram-type matrix~\eqref{e-Gram} that captures the commutativity relations of the defining transvections. 
Viewed as an adjacency matrix, it yields a graph whose directed paths completely determine the given symplectic matrix; see Theorem~\ref{T-T}. These directed paths can be counted with an invertible upper-triangular matrix~\eqref{e-B}, and this allows us to reduce the decomposition problem to a matrix triangulation problem over the binary field. For the latter we make use of the results of~\cite{Botha97}.

In~\cite{PRTC20}, the authors studied the Clifford group via the support~\eqref{e-supp1} of a unitary matrix. 
In that language, Clifford transvections are precisely those Cliffords that have a support of size two, which is smallest support among non-Pauli Cliffords.
On top of being a useful algebraic tool, the support of a unitary encodes valuable information about the Paulis that commute with the given unitary.
In~\cite[Prop.~9]{PRTC20}, the authors compute the support of standard Clifford gates~\eqref{e-GDP}-\eqref{e-GO}. The results of this paper provide a fast algorithm for computing the support of~\emph{any} Clifford gate.
Heuristically, we expect our results to have applications in designing flag gadgets~\cite{Chao-arxiv17a,Chao-arxiv19} for stabilizer circuits.

%%%%%%%%%%%%%%%%%%%%%%%%%%%%%%%%%%%%%%%%%%%%%%%%%%%%%%%%%%%%%%%%%%%%

\section{Preliminaries}
%In this section we will set up notations and layout the relevant terminology.

\subsection{The binary symplectic group}
The binary symplectic group, denoted $\Sp(2m;2)$, consists of $2m\times 2m$ matrices over the binary field $\F_2$ that preserve 
the %\emph{alternating}
matrix
\begin{equation}
    \mathbf{\Omega} = \fourmat{\textbf{0}_m}{\bI_m}{\bI_m}{\textbf{0}_m},
\end{equation}
under congruence. 
That is, $\bF\in \Sp(2m;2)$ iff $\bF\mathbf{\Omega}\bF\T = \mathbf{\Omega}$. Equivalently, symplectic matrices are precisely those matrices that preserve the \emph{symplectic inner product} over $\Fm$
%(induced by $\mathbf{\Omega}$)
\begin{equation}\label{e-sip}    
\inners{(\ba,\bb)}{(\bc,\bd)} = \ba\bd\T +\bb\bc\T = (\ba,\bb)\mathbf{\Omega}(\bc,\bd)\T.
\end{equation}
We will denote by $\GL(n;2)$ and $\Sym(n;2)$ the groups of $n\times n$ invertible and symmetric matrices over the binary field $\F_2$, respectively. 
A matrix $\bF=\fourmat{\bA}{\bB}{\bC}{\bD}\in\Sp(2m;2)$ satisfies $\bF\mathbf{\Omega}\bF\T = \mathbf{\Omega}$, which in turn is equivalent with $\bA\bB\T, \bC\bD\T \in \Sym(m;2)$ and $\bA\bD\T + \bB\bC\T = \bI_m$.
In $\Sp(2m;2)$ we distinguish two subgroups:
\begin{align}
    \cF_D:= & \left\{\bF_D(\bP) = \fourmat{\bP}{\textbf{0}_m}{\textbf{0}_m}{\bP^{- \sf T}}\,\,\middle| \,\,\bP\in \GL(m;2)\right\}, \label{e-FD}\\
    \cF_U:= & \left\{\bF_U(\bS) = \fourmat{\bI_m}{\bS}{\textbf{0}_m}{\bI_m}\,\,\middle| \,\,\bS\in \Sym(m;2)\right\}.\label{e-FU}
\end{align}
Above, $(\sbt)^{- \sf T}$ denotes the inverse transposed, and directly by definition we have $\cF_D \cong \GL(m;2) $ and $\cF_U \cong \Sym(m;2)$. 
Together with matrices
\begin{equation}\label{e-Or}
    \bF_{\mathbf{\Omega}}(r) = \fourmat{\Imrr}{\Imr}{\Imr}{\Imrr},
\end{equation}
with $\Imr$ being the block matrix with $\bI_r$ in upper left corner and 0 elsewhere, and $\Imrr = \bI_m - \Imr$, 
these two groups are the building blocks of the \emph{Bruhat decomposition} with many applications in quantum computation~\cite{MR18,PTC20}.
A symplectic matrix $\bF\in \Sp(2m;2)$ is said to be an \emph{involution} if $\bF^2 = \bI_{2m}$ and is said to be \emph{hyperbolic} if $\inners{\bv}{\bv\bF} = 0$ for all $\bv\in \Fm$. It is straightforward to verify that a hyperbolic map is also an involution.
We will denote
\begin{align}
    \Fix(\bF) & := \ker(\bI +  \bF) := \{\bv\in \Fm\mid \bv = \bv\bF\},\label{e-fix}\\
    \Res(\bF) & := \rs(\bI +  \bF) := \{\bv +  \bv\bF \mid \bv\in \Fm\},\label{e-res}
\end{align}
where $\ker(\sbt)$ and $\rs(\sbt)$ denote the null space and the row space of a matrix, respectively. By definition, these spaces satisfy
\begin{equation}
    \dim \Res(\bF) + \dim \Fix(\bF) = 2m.
\end{equation}
Involutions have the nice property that $\Res(\bF)\subseteq \Fix(\bF)$.
Additionally, for an involution we have $\inners{\bx}{\by\bF} = \inners{\bx\bF}{\by}$ and thus $\inners{\bx+\bx\bF}{\by+\by\bF}= 0$ for all $\bx,\by\in \Fm$. 
%It follows that $\Res(\bF) \subseteq \Res(\bF)^\sperp$, where $(\sbt)^\sperp$ denotes the dual with respect to~\eqref{e-sip}, that is, $\Res(\bF)$ is \emph{self-orthogonal}.
This means that $\Res(\bF)$ is \emph{self-orthogonal} (or \emph{self-dual} if $\dim\Res(\bF) = m$) with respect to~\eqref{e-sip}.

\subsection{The Heisenberg-Weyl group}
The \emph{bit-flip} and the \emph{phase-flip} gates are given by 
\begin{equation}
    \bX := \fourmat{0}{1}{1}{0} \text{ and } \bZ:=\fourmat{1}{0}{0}{-1},
\end{equation}
respectively. 
For vectors $\ba, \mathbf{b} \in \F_2^m$ we will denote 
\begin{equation}
    \bD(\ba,\mathbf{b}):= \bX^{a_1}\bZ^{b_1}\otimes\cdots\otimes\bX^{a_m}\bZ^{b_m}.
\end{equation}
The \emph{Heisenberg-Weyl} group is defined as
\begin{equation}
    \cH\cW_N := \{i^k\bD(\ba,\bb)\mid \ba,\bb\in \F_2^m, k \in \Z_4\}\subset \U(N),
\end{equation}
where $N = 2^m$. 
We will denote by $\cP\cH\cW_N:=\cH\cW_N/\{\pm\bI_N,\pm i\bI_N\}$ the ~\emph{projective} Heisenberg-Weyl group.
Hermitian elements of $\cH\cW_N$ are given (and denoted) by $\bE(\ba,\bb):= i^{\ba\bb^{\sf T}}\bD(\ba,\bb)$.

\subsection{The Clifford group}

The \emph{Clifford group} $\cl_N$ is defined to be the normalizer of $\cH\cW_N$ in $\U(N)$, that is, 
\begin{equation}
\label{e-cliff1}
    \cl_N := \{\bG\in\U(N)\mid \bG \cH\cW_N\bG^\dagger \subset \cH\cW_N\}.
\end{equation}
In order to obtain a finite group,
%(and because phases are irrelevant),
~\eqref{e-cliff1} is meant modulo $\U(1)$.

Let $\{\be_1,\ldots,\be_{2m}\}$ be the standard basis of $\Fm$, and consider $\bG\in \cl_N$. Let $\bc_i\in \Fm$ be such that 
\begin{equation}\label{e-GEG}
    \bG\bE(\be_i)\bG^\dagger = \pm\bE(\bc_i).
\end{equation}
Then the matrix $\bF_{\bG}$ whose $i$th row is $\bc_i$ is a symplectic matrix such that
\begin{equation}\label{e-cliff}
    \bG\bE(\bc)\bG^\dagger = \pm\bE(\bc\bF_{\bG})
\end{equation}
for all $\bc\in \Fm$.
%Based on~\eqref{e-cliff} we obtain 
We thus have a group homomorphism 
\begin{equation}\label{e-Phi}
    \Phi : \cl_N\longrightarrow \Sp(2m;2),\, \bG\longmapsto \bF_{\bG}.
\end{equation}
In addition, $\Phi$ is surjective with kernel $\ker \Phi = \cP\HW$~\cite{RCKP18}, and thus $\cl_N/\cP\cH\cW_N \cong \Sp(2m;2)$. It follows that $\cl_N$ is generated by preimages of symplectic matrices \eqref{e-FD},\eqref{e-FU},\eqref{e-Or}. Here a preimage $\Phi^{-1}(\bF)$ is meant up to $\cH\cW_N$.
%(and up to an overall phase which we have disregarded throughout).
These preimages are, respectively, 
\begin{align}
    \bG_D(\bP) & := |\bv\r\longmapsto |\bv\bP\r,\label{e-GDP}\\
    \bG_U(\bS) & := \diag\left(i^{{\bv\bS\bv^{\sf T}} \mod 4}\right)_{\bv\in\F_2^m}, \label{e-GUS}\\
    \bG_{\mathbf{\Omega}}(r) & :=(\bH_2)^{\otimes r}\otimes \bI_{2^{m-r}}\label{e-GO},
\end{align}
where $\bH_2$ is the Hadamard gate.

Since $\Phi$ is a homomorphism we have that $\Phi(\bG^\dagger) = \bF_\bG^{-1}$. It follows that if $\bG\in\cl_N$ is Hermitian then $\bF_\bG$ is a symplectic involution. Conversely, if $\bF$ is a symplectic involution then $\bG = \Phi^{-1}(\bF)$ satisfies $\bG^2 \in \cH\cW_N$. 
As mentioned, a special class of involutions are the hyperbolic maps. If $\bG\in\cl_N$ corresponds to a hyperbolic $\bF \in \Sp(2m;2)$ then~\eqref{e-cliff} implies that $\bG\bE\bG^\dagger$ commutes with $\bE$ for all $\bE$.
%$\bG\bE\bG\bE =  \bE\bG\bE\bG$ for all $\bE$. 
%%%%%%%%%%%%%%%%%%%%%%%%%%%%%%%%%%%%%%%%%%%%%%%%%%%%%%%%%%%%%%%%%%%%%%%%%%%%%%%%%%%%%%%%%%%%%%%%%%%%%%%

\section{Transvection Decomposition of Symplectic Matrices}\label{Sec-decomp}

A \emph{symplectic transvection} is a symplectic map with one-dimensional residue space. It is easily seen that if $\Res(\bF) = \l\bv\r$ then the matrix $\bF\in \Sp(2m;2)$ must act as 
%$\bx\bF = \bx + \inners{\bx}{\bv}\bv$. 
%It follows that symplectic transvections are of form
\begin{equation}\label{e-transvectionDef}
    \bT_\bv := \bI + \mathbf{\Omega}\bv\T\bv, \, \bx \longmapsto \bx + \inners{\bx}{\bv}\bv.
\end{equation}
We will call two transvections $\bT_\bv,\bT_\bw$ \emph{independent} if the defining $\bv,\bw$ are independent. Otherwise, we will call the transvections dependent. Note also that $\bT_\bv,\bT_\bw$ \emph{commute}, that is,  $\bT_\bv\cdot\bT_\bw = \bT_\bw\cdot\bT_\bv$ iff $\inners{\bv}{\bw}=0$, that is, iff $\bv,\bw$ are orthogonal (with respect to~\eqref{e-sip} of course).

It is well-known that $\Sp(2m;2)$ is generated by transvections. It is shown in~\cite{OMeara,Callan} that a non-hyperbolic map $\bF$ can be written as a product of $r$ \emph{independent} transvections $\bT_{\bv_1},\ldots,\bT_{\bv_r}$, where $r = r(\bF):= \dim\Res(\bF) = 2m - \dim\Fix(\bF)$ and $\Res(\bF) = \l\bv_1,\ldots,\bv_r\r$. The strategy of~\cite{OMeara} is to find $\bv$ such that $\inners{\bx}{\bx\bF} = 1$ (which exists for non-hyperbolics), and consider $\bF\bT_\bv, \bv = \bx+\bx\bF \in \Res(\bF)$, for which $r(\bF\bT_\bv) = r(\bF) - 1$. One then repeats the process accordingly until a one-dimensional residue space is reached.

The following result will enable us to restrict without loss of generality to non-hyperbolic maps.
\begin{lem}[\mbox{\cite[2.1.8]{OMeara}}]\label{l-hyp}
Let $\bF\in \Sp(2m;2)$ be hyperbolic. Then there exists $\bv\in \Fm$ such that $\bF\bT_\bv$ is non-hyperbolic and $\Res(\bF) = \Res(\bF\bT_\bv)$. 
\end{lem}
\begin{proof}
Fix any $\mathbf{0} \neq \bv = \bx + \bx\bF \in \Res(\bF)$.
Then any $\by$ such that $\inners{\by}{\bv} = 1 = \inners{\by\bF}{\bv}$ (which of course exists) satisfies $\inners{\by}{\by\bF\bT_\bv} = 1$, and thus $\bF\bT_\bv$ is non-hyperbolic. Next, by the choice of $\bv$, $\Res(\bF\bT_\bv) \subseteq \Res(\bF)$ holds trivially, and equality is due to equal cardinalities.
\end{proof}
%A hyperbolic map can be easily transformed to non-hyperbolic. Indeed, let $\bF$ be any hyperbolic and fix any\footnote{In fact any $0\neq \bv \in \Fm$ would still work but we strategically restrict to $\Res(\bF)$.} $0 \neq \bv = \bx + \bx\bF \in \Res(\bF)$. 
%Then any $\by$ such that $\inners{\by}{\bv} = 1 = \inners{\by\bF}{\bv}$ (which of course exists) satisfies $\inners{\by}{\by\bF\bT_\bv} = 1$, and thus $\bF\bT_\bv$ is non-hyperbolic. 
%By the choice of $\bv$ we have that $\bF\bT_\bv = \bT_\bv\bF$ and thus $\bF\bT_\bv$ is also an involution. 
It follows from Lemma~\ref{l-hyp} that a hyperbolic map $\bF$ is a product of $r+ 1$ transvections, $r$ of which form a basis for $\Res(\bF)$, and the additional transvection is dependent of the first $r$.

For involutions (hyperbolic or not) we have the following nicer result.

\begin{prop}\label{P-ComT} Any involution is a product of commuting transvections. The converse is also true, that is, any product of commuting transvections yields and involution.
\end{prop}
\begin{proof}
The result follows immediately by the fact that two transvections commute iff their defining vectors are orthogonal, along with the fact that the residue space of an involution is self-orthogonal.
\end{proof}
%%%%%%%%%%%%%%%%%%%%%%%%%%%%%%%%%%%%%%%%%%%%%%%%%%%%%%%%%%%%%%%%%%%%%%%%%%%%%%%%%%%%%%%%%%%
\subsection{A Gram-type matrix}

In this section $\bF$ will be a generic symplectic matrix. We associate to a minimal transvection decomposition $\bF = \bT_{\bv_1}\cdots\bT_{\bv_r}$ a Gram-type matrix
\begin{equation}\label{e-Gram}
    \bA(\bv_1,\ldots,\bv_r) := [\inners{\bv_i}{\bv_j}]_{i,j} = \bV\mathbf{\Omega}\bV\T,
\end{equation}
where $\bV$ is the $r\times 2m$ matrix formed by stacking $\bv_1,\ldots,\bv_r$. Obviously, $\bA$ is symmetric and has zero diagonal. 
%For our purposes it suffices to focus only in the case when $\bV$ is a basis of $\Res(\bF)$. 
Since a minimal transvection decomposition is given by some basis of the residue space, we will assume that $\bv_i \in \Res(\bF)$.
Note that $\bA = \mathbf{0}$ iff $\bF$ is an involution iff $\bV$ is self-orthogonal. 
On the other hand, 
\begin{align}
    \inners{\bv_i}{\bv_j} & = \inners{\bx_i + \bx_i\bF}{\bx_j + \bx_j\bF} \\
    & = \inners{\bx_i\bF}{\bx_j} + \inners{\bx_i}{\bx_j\bF} \\
    & = \bx_i\bF\mathbf{\Omega}\bx_j^{\sf T}+ \bx_j\mathbf{\Omega}\bF\T\bx_j^{\sf T} \\
    & = \bx_i(\bF+\bF^{-1})\mathbf{\Omega}\bx_j^{\sf T} \\
    & = \inners{\bx_i(\bF+\bF^{-1})}{\bx_j}.
\end{align}
Obviously, $\bF$ is an involution iff $\bF = \bF^{-1}$, and thus $\bA$ also captures how far is $\bF$ from being an involution, or equivalently, how far is $\bV$ from being self-orthogonal.
In what follows we will denote $\bA_{\rm u}:= \texttt{triu}(\bA)$ the upper triangular part of $\bA$ and 
\begin{equation}\label{e-B}
    \bB(\bv_1,\ldots,\bv_r) := \sum_{\ell = 0}^{r-1} \bA_{\rm u}^\ell.
\end{equation}
By definition, it follows that $\bB$ is upper triangular with all-ones diagonal for any symplectic $\bF$, and is the identity matrix for any involution (since in this case $\bA = \mathbf{0})$. Moreover, $\bA_{\rm u}$ is $r\times r$ upper triangular with all-zero diagonal. This yields $\bA_{\rm u}^r = \mathbf{0}$, and thus 
\begin{equation}\label{e-B1}
    \bB = (\bI_r + \bA_{\rm u})^{-1} \text{ and } \bA_{\rm u} = \bI_r + \bB^{-1}.
\end{equation}
The matrices $\bA$ and $\bB$ have a natural graphical interpretation. Let us start with $\bA$, which can be thought as the adjacency matrix of the graph with vertices $\bv_i$ and edges $(\bv_i,\bv_j)$ iff $\inners{\bv_i}{\bv_j} = 1$. 
On the other hand, its upper triangular part $\bA_{\rm u}$ can be thought as the adjacency matrix of the corresponding \emph{directed} graph with edges $(\bv_i,\bv_j)$ iff $\inners{\bv_i}{\bv_j} = 1$ \emph{and} $i < j$. 
As for the matrix $\bB$, note first that entry $(i,j)$ of $\bA_{\rm u}^\ell$ counts directed paths from $\bv_i$ to $\bv_j$ of length $\ell$. Thus, entry $(i,j)$ (always for $i < j$) counts the number of directed paths from $\bv_i$ to $\bv_j$. 

Before providing an example of the notions introduced, we point out that the matrix $\bB$ also captures the number of \emph{distinct} transvection decompositions of a given symplectic matrix $\bF$. However, this treatment goes beyond the scope of this paper and will be presented in future work.

\begin{exa}
Let us consider an example with $m=5$ and $\bF = \bT_{\bv_1}\bT_{\bv_2}\bT_{\bv_3}\bT_{\bv_4}\bT_{\bv_5}$, where 
\[
\bV = \left[\!\!\begin{array}{c}
     \bv_1  \\
     \bv_2  \\
     \bv_3  \\
     \bv_4  \\
     \bv_5  \\
\end{array}
\!\!\right] = \left[\!\!\begin{array}{cccccccccc}
     1&0&0&1&1&1&0&0&1&0  \\
     1&1&0&0&0&1&1&1&1&0  \\
     0&1&1&0&1&1&0&0&1&0  \\
     0&0&0&1&0&1&1&1&1&0  \\
     0&1&0&0&0&1&0&0&0&1
\end{array}
\!\!\right].
\]
Then one computes
\[
\bA = 
%\bV\mathbf{\Omega}\bV\T = 
\left[\!\!\begin{array}{ccccc}
     0&1&0&1&0  \\
     1&0&1&1&0  \\
     0&1&0&1&1  \\
     1&1&1&0&1  \\
     0&0&1&1&0  \\
\end{array}
\!\!\right]
\text{ and }
\bB = 
%\sum_{\ell = 0}^{4} \bA_{\rm u}^\ell = 
\left[\!\!\begin{array}{ccccc}
     1&1&1&3&4  \\
     0&1&1&2&3  \\
     0&0&1&1&2  \\
     0&0&0&1&1  \\
     0&0&0&0&1  \\
\end{array}
\!\!\right].
\]
The graphical description of this scenario is given in Figure~\ref{F-dg}. For instance, entry $b_{1,4} = 3$ and there are precisely three directed paths from $\bv_1$ to $\bv_4$, namely, $(\bv_1,\bv_4)$, $(\bv_1,\bv_2,\bv_4)$, and $(\bv_1,\bv_2,\bv_3,\bv_4)$. 
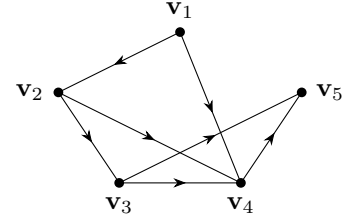
\begin{figure}[!h]
\centering
\begin{tikzpicture}[scale=.8,colorstyle/.style={circle, draw=black!100,fill=black!100, thick, inner sep=0pt, minimum size=1mm},every node/.style={sloped,allow upside down}]
    \node at (0, 1)[colorstyle,label=above:$\bv_1$]{};
    \node at (-2, 0)[colorstyle,label=left:$\bv_2$]{};
    \node at (-1,-1.5)[colorstyle,label=below:$\bv_3$]{};
    \node at (1,-1.5)[colorstyle,label=below:$\bv_4$]{};
    \node at (2, 0)[colorstyle,label=right:$\bv_5$]{};
    \draw(0,1) -- node {\midarrow} (-2,0);
    \draw(0,1) -- node {\midarrow}(1,-1.5);
    \draw(-2,0) -- node {\midarrow}(-1,-1.5);
    \draw(-2,0) -- node {\midarrow}(1,-1.5);
    \draw(-1,-1.5) -- node {\midarrow}(1,-1.5);
    \draw(-1,-1.5) -- node {\midarrow}(2,0);
    \draw(1,-1.5) -- node {\midarrow}(2,0);
\end{tikzpicture}
\caption{The directed graph with adjacency matrix $\bA_{\rm u}$.}
\label{F-dg}
\end{figure}
\end{exa}

%%%%%%%%%%%%%%%%%%%%%%%%%%%%%%%%%%%%%%%%%%%%%%%%%%%%%%%

\begin{theo}\label{T-T}
For any symplectic matrix $\bF = \bT_{\bv_1}\cdots\bT_{\bv_r}$ 
%and for any $\bx\in \Fm$ 
we have
\begin{equation}
    \bF = \bI + \mathbf{\Omega} \bV\T \bB\bV.
\end{equation}
As a consequence, if $\bF$ is an involution then $\bF = \bI + \mathbf{\Omega} \bV\T \bV$
\end{theo}
\begin{proof}
By the definition of transvections, the action of $\bF$ on $\bx$ is given by some linear combination of $\bv_j$ added to $\bx$, that is 
\begin{equation}
    \bx\bF = \bx + \sum_{j = 1}^r w_j\bv_j,
\end{equation}
where $w_j$ depends on $\inners{\bx}{\bv_i}$ for $i<j$. We claim that $w_j = \bx\mathbf{\Omega}\bV\T \bB_j$ where $\bB_j$ is the $j$th column of $\bB$. This in turn will complete the proof. In order to prove the claim, note that the input of $\bT_{\bv_j}$ is $\bx\bT_{\bv_1}\cdots \bT_{\bv_{j-1}}$. Thus $\inners{\bx}{\bv_i}$ contributes to $w_j$ only if $i<j$ \emph{and} there is a directed path form $\bv_i$ to $\bv_j$, which could be of length $1\leq \ell \leq j-i$. This information is precisely encoded by $\bB_j$.
\end{proof}

To the best of our knowledge, Theorem~\ref{T-T} constitutes a novel structural result about symplectic matrices, and comparing it with~\eqref{e-transvectionDef}, should come as no surprise. This structure is the main building block of what follows. 
Based on Theorem~\ref{T-T}, it is imperative to consider the \emph{residue matrix} 
\begin{equation}\label{e-Fhat}
    \widehat{\bF} := \mathbf{\Omega}(\bI + \bF) = \bV\T \bB \bV = \sum_{i,j} b_{i,j}\bv_i\!\!\T\,\bv_j.
\end{equation}
The terminology comes from the obvious fact that $\rs (\widehat{\bF}) = \Res(\bF)$. Note that $\widehat{\bF}$ is symmetric iff $\bB = \bI$ (recall that $\bB$ is lower triangular) iff $\bF$ is an involution. 
Moreover, since $\widehat{\bF\T}$ has all-zero diagonal iff $\widehat{\bF}$ does, and since
\begin{equation}
    \bx\widehat{\bF\T}\bx\T = \bx\mathbf{\Omega}(\bI + \bF\T)\bx\T = \bx\mathbf{\Omega}\bx\T + \bx\mathbf{\Omega}\bF\T\bx\T = \inners{\bx}{\bx\bF},
\end{equation}
we conclude that $\widehat{\bF}$ has all-zero diagonal iff $\bF$ is hyperbolic. In such case $\bF$ is also an involution, and thus $\widehat{\bF}$ is \emph{alternating} (that is, symmetric and all-zero diagonal). It follows by Lemma~\ref{l-hyp} that we may restrict ourselves to non-hyperbolic maps, and thus we will assume that $\widehat{\bF}$ \emph{is not} alternating. 
%Under this assumption we have the following.
%%%%%%%%%%%%%%%%%%%%%%%%%%%%%%%%%%%%%%%%%%%%%%%%%%%%%%%%%%%%%%%%%%%%
\subsection{Decomposition of Symplectic Involutions}
In this subsection we will present a simple algorithm for the decomposition of (non-hyperbolic)  symplectic involutions, and provide intuition for the much more delicate decomposition of general symplectic matrices. 

\begin{theo}[Transvection Decomposition of Involutions]\label{T-main}
Let $\bF$ be a non-hyperbolic involution, so that the residue matrix $\widehat{\bF}$ is non-alternating.
Then there exists $\bP\in \GL(2m;2)$ such that $\bF = \bT_{\bv_1}\cdots\bT_{\bv_r}$, where $r = \dim\Res(\bF)$ and $\bv_j$ is the $j$th row of $\bP\widehat{\bF}$ for $1\leq j \leq r$.
\end{theo}
\begin{proof}
Let $\bR$ be the matrix of row operations that transforms $\widehat{\bF}$ into Row-Reduced Echelon form. Let $\bE$ be the $r\times r$ upper left block of $\bR\widehat{\bF}\bR\T$, which is invertible by construction. 
It will also be symmetric and have non-zero diagonal since $\bF$ is non-hyperbolic involution. 
Then there exists $\bQ \in \GL(r;2)$ such that $\bQ\bE\bQ\T = \bI_r$; see~\cite[2.1.14]{OMeara} for instance. Now put $\bP = \texttt{blkdiag}(\bQ,\bI_{2m-r}) \bR$. Then 
\begin{equation}\label{e-T}
    \bP\widehat{\bF}\bP\T = \fourmat{\bI_r}{0}{0}{0}.
\end{equation}
We will consider the nonzero rows of $\bP\widehat{\bF}$, that is, $\twomat{\bQ\bE}{\bf 0}\bR^{- \sf T}$. For $1\leq j \leq r$ let $\bw_j$ denote the $j$th row of $\bP\widehat{\bF}$, that is, $\bw_j = \be_j\bP^{-\sf T}$, where $\be_j \in \Fm$ is the $j$th standard basis vector. 
Put $\bF' = \bT_{\bw_1}\cdots\bT_{\bw_r}$. Since $\bw_j$'s are linear combinations of $\bv_j$'s and since $\bF$ is an involution it follows that $\bA(\bw_1,\ldots,\bw_r) = \bA(\bv_1,\ldots,\bv_r) = \mathbf{0}_r$ and $\bB(\bw_1,\ldots,\bw_r) = \bB(\bv_1,\ldots,\bv_r) = \bI_r$. Then~\eqref{e-Fhat} yields
\begin{equation}
    \widehat{\bF'} = \sum_{j = 1}^r \bw_j\!\!\T\bw_j = \sum_{j = 1}^r \bP^{-1}\be_j\!\!\T\be_j\bP^{-\sf T} = \bP^{-1} \fourmat{\bI_r}{0}{0}{0} \bP^{-\sf T} = \widehat{\bF},
\end{equation}
and thus $\bF = \bF'$.
\end{proof}

The strength of Theorem~\ref{T-main} is that, as we will see, it can be generalized to non-involutions. The case of involutions can be dealt separately with an alternate approach, which, however, does not generalize to non-involutions. According to~\cite[Thm.~4.1]{Beigi-qic10}, an involution $\bF$ is conjugate with an involution of form
\begin{equation}\label{e-FUS}
    \bF_U(\bS) \equiv \fourmat{\bI}{\bS}{\mathbf{0}}{\bI}, \quad \bS\in \Sym(m;2),
\end{equation}
that is, there exists $\bM \in \Sp(2m;2)$ such that $\bM\bF\bM^{-1} = \bF_U(\bS)$. On the other hand, the involutions of form~\eqref{e-FUS} are easy to decompose as described in~\cite[Prop.~9(2)]{PRTC20}. So let us assume $\bF_U(\bS) = \bT_{\bv_1}\cdots \bT_{\bv_r}$. It is straightforward to verify that $\bT_\bv\bM = \bM\bT_{\bv\bM}$ holds for any symplectic $\bM$. This yields
\begin{align}
    \bF & = \bM^{-1}\bF_U(\bS)\bM \\
    & = (\bM^{-1}\bT_{\bv_1}\bM)\cdot (\bM^{-1}\bT_{\bv_2}\bM)\cdots  (\bM^{-1}\bT_{\bv_r}\bM) \\
    & =\bT_{\bv_1\bM}\cdots \bT_{\bv_r\bM}.
\end{align}
%%%%%%%%%%%%%%%%%%%%%%%%%%%%%%%%%%%%%%%%%%%%%%%%%%%%%%%%%%%%%%%%%%%%%%%%%%%%%%
\subsection{Decomposition of Symplectic Matrices}
Finding a transvection decomposition for involutions is facilitated by the simple nature of their associated $\bA$ and $\bB$ matrices. As we will see, the general case is much more complicated. Let $\bF$ be any (non-hyperbolic) symplectic matrix and consider its residue matrix $\widehat{\bF}$, for which $\rk(\widehat{\bF}) = r$. 
Thus, a transvection decomposition of $\bF$ is given by some basis of $\Res(\bF) = \rs(\widehat{\bF})$. 
The task in hand is how to find such basis. The main idea is to start with some fixed basis 
%, as in Theorem~\ref{T-main}, 
and transform it accordingly until we reach the desired result. 
We will start with a basis of $\Res(\bF)$ in Row-Reduced Echelon form, that is, let $\bR$ be a matrix of row operations so that $\bR\widehat{\bF} = \twomatv{\bV}{\bf 0}$, where $\bV$ is a $r\times 2m$ basis. This can be done, for instance, via Gauss Elimination over $\F_2$. Then 
\begin{equation}\label{e-rref}
\bR\widehat{\bF}\bR\T = \fourmat{\bE}{\bf 0}{\bf 0}{\bf 0}, \, \bE\in \GL(r;2).
\end{equation}
As mentioned, the basis $\bV$ may or may not constitute a transvection decomposition of $\bF$, and the idea is to consider other bases of form $\bQ\bV$ where $\bQ \in \GL(r;2)$. Let us denote $\bP = \texttt{blkdiag}(\bQ,\bI_{2m-r})$, and let $\bB = \bB(\bQ\bV)$.
%%%%%%%%%%%%%%%%%%%%%%%%%%%%%%%%
\begin{lem}\label{l-TET}
With the same notation as above, the basis $\bQ\bV$ constitutes a transvection decomposition of $\bF$ iff $\bQ\bE\bQ\T = \bB^{- \sf T}$.
\end{lem}
\begin{proof}
Assume $\bQ\bV$ gives a transvection decomposition for $\bF$. Then, by~\eqref{e-Fhat} we have $\widehat{\bF} = (\bQ\bV)\T \cdot \bB \cdot  (\bQ\bV)$. But with the notation above we have $\bQ\bV = \twomat{\bQ\bE}{\bf 0}\bR^{- \sf T}$. Thus
\begin{align}
    \widehat{\bF}  & = \bV\T\bQ\T \cdot \bB \cdot \bQ\bV \\
    & = \bR^{-1}\twomatv{\bE\T\bQ\T}{\bf 0} \, \cdot \bB \cdot \, \twomat{\bQ\bE}{\bf 0} \bR^{- \sf T}\\
    & = \bR^{-1}\fourmat{\bE\T\bQ\T \cdot \bB \cdot \bQ\bE}{\bf 0}{\bf 0}{\bf 0} \bR^{- \sf T}.
\end{align}
It follows by~\eqref{e-rref} that $\bE = \bE\T\bQ\T \cdot \bB \cdot \bQ\bE$ and thus $\bQ\bE\bQ\T = \bB^{-\sf T}$. The reverse direction follows similarly.
\end{proof}
% LEMMA 3 BEGIN
\begin{lem}\label{l-QEQ}
With the same notation as above, if $\bQ\bE\bQ\T$ is
%is invertible and 
lower triangular, then $\bQ\bE\bQ\T = \bB^{-\sf T}$.
\end{lem}
\begin{proof}
Assume $\bQ\bE\bQ\T$ is 
%invertible and 
lower triangular and put $\bE' = \bE\T\bQ\T \cdot \bB \cdot \bQ\bE$. Then
\begin{align}
    %\bQ\bE'\bQ\T = \bQ\bE\T\bQ\T \cdot \bB \cdot \bQ\bE\bQ\T, \\
    \bQ\bE\bQ\T = ((\bQ\bE\bQ\T)\T \cdot \bB)^{-1} \cdot \bQ\bE'\bQ\T.
\end{align}
If $\bQ\bE\bQ\T = \bI_r$, the statement is clear because in this case $\widehat{\bF}$ is symmetric, and therefore $\bB = \bI_r$.
If $\bQ\bE\bQ\T \neq \bI_r$, then both $\bQ\bE'\bQ\T$ and $(\bQ\bE\bQ\T)\T \cdot \bB$ have to be invertible and lower triangular. But $(\bQ\bE\bQ\T)\T$ and $\bB$ are both invertible and upper triangular, meaning $(\bQ\bE\bQ\T)\T \cdot \bB$ is also upper triangular.
Thus $(\bQ\bE\bQ\T)\T \cdot \bB = \bI_r$, and $\bQ\bE\bQ\T = \bB^{-\sf T}$ 
%as stated
%\begin{align}
 %   (\bQ\bE\bQ\T)\T \cdot \bB = \bI_r, \\
%    \bQ\bE\bQ\T = \bB^{-\sf T}.
%\end{align}
\end{proof}
% LEMMA 3 END
It follows by Lemmas~\ref{l-TET} and~\ref{l-QEQ} that we are seeking for matrices $\bQ$ that triangularize $\bE$ from~\eqref{e-rref} by congruence. 
For more on triangularizations by congruence and related algorithms we refer the reader to~\cite{Botha97}. 
It also follows by Lemma~\ref{l-TET} that $\widehat{\bF}$ \emph{can} be triangularized by congruence for any non-hyperbolic $\bF$ (since for this, one would only need a transvection decomposition of $\bF$, which we know it always exists). 
We resume everything to the following theorem.
\begin{theo}[Transvection Decomposition of Symplectic Matrices]\label{T-main1}
Let $\bF$ be a generic symplectic matrix. Then there exists an algorithm that for any generic symplectic matrix $\bF$ outputs a minimal transvection decomposition.
\end{theo}
\begin{proof}
If the residue matrix $\widehat{\bF}$ is alternating, that is, if $\bF$ is hyperbolic, then pick $\bv$ as in Lemma~\ref{l-hyp} and update the input $\bF$ with the non-hyperbolic $\bF\bT_\bv$, while keeping the residue space intact. Next, perform Gauss Elimination on $\widehat{\bF}$ with $\bR$ as in~\eqref{e-rref}, and let $\bQ$ be such that $\bQ\bE\bQ\T$ is lower triangular. Then, by Lemmas~\ref{l-TET} and ~\ref{l-QEQ}, the $r$ nonzero rows of $\texttt{blkdiag}(\bQ,\bI_{2m-r})\bR\widehat{\bF}$, where $r = \dim \Res(\bF)$, along with $\bv$, yield a minimal transvection decomposition for $\bF$.
\end{proof}

%We resume everything in Algorithm~\ref{alg}.
%\begin{algorithm} \caption{Decomposition of Symplectic Matrices} \label{alg}
%{\bf Input:} A symplectic matrix $\bF$. 
%\begin{algorithmic}
%\STATE~1.  \hspace{.02 in} Compute $\widehat{\bF} = \mathbf{\Omega}(\bI + \bF)$.
%\STATE~2. \hspace{.02 in} \textbf{if } $\widehat{\bF}$ is alternating \textbf{do}
%\STATE~3. \hspace{.02 in} \quad $\bv$ = first non-zero row of $\widehat{\bF}$ 
%\STATE~4. \hspace{.02 in}  \quad $\bF \leftarrow \bF\bT_\bv$.
%\STATE~5. \hspace{.02 in}  \quad $\widehat{\bF} \leftarrow \mathbf{\Omega}(\bI + \bF)$.
%\STATE~6. \hspace{.02 in} \textbf{end if}
%\STATE~7. \hspace{.02 in} $r = \rk(\widehat{\bF})$.
%\STATE~8. \hspace{.02 in} Perform Gauss Elimination on $\widehat{\bF}$ with $\bR$ as in~\eqref{e-rref}.
%\STATE~9. \hspace{.02 in} Triangularize $\bE$ from~\eqref{e-rref} with $\bQ$.
%\STATE~10. $\bV \leftarrow$ first $r$ rows of $\texttt{blkdiag}(\bQ,\bI_{2m-r})\cdot\bR\cdot\widehat{\bF}$.
%\end{algorithmic}
%{\bf Output:} $\bv,\,\bV$.
%\end{algorithm}

%%%%%%%%%%%%%%%%%%%%%%%%%%%%%%%%%%%%%%%%%%%%%%%%%%%%%%%%%%%%%%%%
\section{Decomposition of Clifford Gates}
In~\cite{PRTC20}, the authors studied the Clifford hierarchy via the support of the underlying gates. Every gate $\bU \in \U(N)$ can be written as
\begin{equation}\label{e-supp}
    \bU = \frac{1}{N}\sum_{\bv\in \Fm}\Tr\big(\bE(\bv)\bU\big)\,\bE(\bv),
\end{equation}
and the support of $\bU$ consist of the basis terms that appear in~\eqref{e-supp}, that is,
\begin{equation}\label{e-supp1}
    \supp(\bU) := \{\bE(\bv)\in \cH\cW_N \mid \Tr\big(\bE(\bv)\bU\big) \neq 0\}.
\end{equation}
Given the isomorphism $\bE(\bv) \longleftrightarrow \bv$, the support can be equivalently though of as a subspace of $\Fm$. On the other hand,~\eqref{e-Phi} assigns $\bF\in \Sp(2m;2)$ to a coset $\cH\cW_N \bG = \Phi^{-1}(\bF)$ for any $\bG \in\cl_N$. 
%We will denote such coset by $\Phi^{-1}(\bF)$. 
It is straightforward to verify that the Clifford 
\begin{align}\label{e-transvection}
\bG_\bv:= \frac{\bI_N \pm i\bE(\bv)}{\sqrt{2}} \in \cl_N
\end{align}
corresponds to the transvection $\bT_\bv$.
Then, since every symplectic is a product of transvections, it follows that
\begin{equation}\label{e-GT}
    \bG = \bE_0\prod_{n = 1}^k\frac{\bI_N + i\bE_n}{\sqrt{2}} = \frac{\bE_0}{\sqrt{|S|}}\sum_{\bE\in S}\alpha_\bE\bE,
\end{equation}
where $\bE_0\in \cH\cW_N, S = \l\bE_1,\ldots,\bE_k\r$, and $\alpha_\bE \in \C$; see ~\cite[Prop.~4]{PRTC20}. 
From earlier discussion, it follows that the support of any $\bG\in \Phi^{-1}(\bF)$ is given by $\Res(\bF)$ if $\bF$ is non-hyperbolic, and by some subspace of $\Res(\bF)$ of index $2$ otherwise. In~\cite[Prop.~9]{PRTC20}, the authors determined the support of the standard Clifford gates~\eqref{e-GDP}-\eqref{e-GO}, while the general case remained open.
The difficulty arose by the fact that the support of products is hard to compute. This problem can now be solved with the aid of Theorem~\ref{T-main1}, as resumed in Algorithm~\ref{alg1}. 
\begin{algorithm} \caption{Transvection Decomposition of Clifford Gates} \label{alg1}
{\bf Input:} A Clifford gate $\bG$. 
\begin{algorithmic}
\STATE~1.   Compute $\bF$ from~\eqref{e-GEG}.
\STATE~2.  Compute $\bv,\bv_1,\cdots,\bv_r$ from Theorem~\ref{T-main1}.
\STATE~3.  $\bG_0 = \bG_\bv\prod_j\bG_{\bv_j} $.
\STATE~4.  Find $\bE_0 = \bE(\bv_0)$ such that $\bG = \bE_0\bG_0$.
\end{algorithmic}
{\bf Output:} $\bv_0,\,\bv_j$'s.
\end{algorithm}
It is also worth mentioning that in this process one may lose an eighth root of unity; see Example~\ref{ex-CNOT} for instance.
We point out here that $\bG$ is traceless iff $\bE_0\notin S$. Thus the search in Step 4. of Algorithm~\ref{alg1} can be reduced to either outside $S$ if $\bG$ is traceless or in $S$ otherwise.
%%%%%%%%%%%%%%%%%%%%%%%%%%%%%%%%%%%%%%%%%%%%%%%%
\begin{exa}
The Hadamard gate can be written as
\begin{equation}
    \bH_2 = \frac{1}{\sqrt{2}}(\bX + \bZ) = \bX \frac{\bI + i\bY}{\sqrt{2}},
\end{equation}
where $\bY = i\bX\bZ$ as usual. Consider now the $m$ fold transversal Hadamard gate $\bH_N = (\bH_2)^{\otimes m}$, for which $\Phi(\bH_N) = \mathbf{\Omega}$. Additionally $\widehat{\mathbf{\Omega}} = \fourmat{\bI}{\bI}{\bI}{\bI}$ and $\dim \Res(\mathbf{\Omega}) = m$. Then $\bP = \fourmat{\bI}{\mathbf 0}{\bI}{\bI}$ triangularizes $\widehat{\mathbf{\Omega}}$:
\begin{equation}
    \bP\widehat{\mathbf{\Omega}}\bP\T = \fourmat{\bI}{\mathbf 0}{\mathbf 0}{\mathbf 0}.
\end{equation}
The first $m$ nonzero rows of $\bP\widehat{\mathbf{\Omega}}$ are $\twomat{\bI}{\bI}$. We see that the $n$th row yields the gate $\bY_n$ with $\bY$ in qubit $n$ and identity elsewhere. From Step 3. of Algorithm~\ref{alg1} we compute
\begin{equation}
    \bG_0 = \prod_{n = 1}^m\frac{\bI_N + i\bY_n}{\sqrt{2}}. 
\end{equation}
We then find $\bH_N = \bX^{\otimes m}\bG_0$. A similar result holds for partial Hadamard gates $\bH^{\otimes r}\otimes \bI_{2^{m-r}}$, to which correspond symplectics of form~\eqref{e-Or}; see also~\cite[Prop.~9(3)]{PRTC20}
\end{exa}
\begin{exa}\label{ex-CNOT}
The symplectic and residue matrices corresponding to the CNOT gate are given by
\begin{equation}
    \bF = \left[\!\!\begin{array}{cccc}
         1&1&0&0 \\
         0&1&0&0 \\
         0&0&1&0 \\
         0&0&1&1
    \end{array}
    \!\!\right]
    \text{ and }
        \widehat{\bF} = \left[\!\!\begin{array}{cccc}
         0&0&0&0 \\
         0&0&1&0 \\
         0&1&0&0 \\
         0&0&0&0
    \end{array}
    \!\!\right],
\end{equation}
from which we see that $\widehat{\bF}$ is alternating, and thus $\bF$ is hyperbolic. So first, we transform $\bF$ to a non-hyperbolic map by using the first non-zero row of $\widehat{\bF}$, that is, $\bv = 0010$. Then we update $\bF \leftarrow \bF\bT_\bv$, for which 
\begin{equation}
    \bF = \left[\!\!\begin{array}{cccc}
         1&1&1&0 \\
         0&1&0&0 \\
         0&0&1&0 \\
         0&0&1&1
    \end{array}
    \!\!\right]
    \text{ and }
        \widehat{\bF} = \left[\!\!\begin{array}{cccc}
         0&0&0&0 \\
         0&0&1&0 \\
         0&1&1&0 \\
         0&0&0&0
    \end{array}
    \!\!\right].
\end{equation}
A matrix that triangularizes $\widehat{\bF}$ is given by
\begin{equation}
        \bP = \left[\!\!\begin{array}{cccc}
         0&1&1&0 \\
         0&0&1&0 \\
         1&0&1&0 \\
         0&0&0&1
    \end{array}
    \!\!\right].
\end{equation}
The non-zero rows of $\bP\widehat{\bF}$ are $\bv_1 = 0100,\bv_2 = 0110$. Note that $\bv = \bv_1 + \bv_2$, and $\inners{\bv_1}{\bv_2} = 0$ as in Proposition~\ref{P-ComT}. Then we compute
\begin{align}
    \bG_0 
    & = \frac{(\bI + i\bI\otimes \bX)(\bI - i\bZ\otimes \bX)(\bI +i\bZ\otimes \bI)}{\sqrt{8}}.
\end{align}
And then we end with the observation that CNOT$ = \xi\bG_0$, where $\xi = (1 - i)/\sqrt{2}$ is an eighth root of unity.
\end{exa}

\section*{Acknowledgements}
This work was funded in part by the Academy of Finland (grant 334539).

%%%%%%%%%%%%%%%%%%%%%%%%%%%%%%%%%%%%%%%%%%%%%%%%%%%%%%%%%%%%%%%%%%%%%%%%%%%%%%%%5
\bibliographystyle{IEEEtran}
\bibliography{IEEEabrv,ITW2020}

\end{document}